\documentclass{article}

\usepackage{amsmath,amssymb,pifont}
\usepackage{multicol}
\usepackage{amstext}
\usepackage{amsthm}
\usepackage{multirow}
\usepackage{booktabs}
\usepackage[skip=0pt]{subcaption}
\usepackage{times}
\usepackage{lipsum}
\usepackage[shortlabels]{enumitem}
\usepackage{cancel}
\usepackage{wrapfig}
\usepackage{array}
\usepackage{siunitx}
\usepackage{csvsimple}
\usepackage[multidot]{grffile}
\usepackage{bbm}
\usepackage{dblfloatfix}
\usepackage{hyperref}
\usepackage{makecell}
\usepackage{bbm, dsfont}
\usepackage{mathtools}
\usepackage[dvipsnames]{xcolor}
\usepackage{comment}
\usepackage[capitalise]{cleveref}

\usepackage{geometry}
\usepackage{mathpazo}
\usepackage{enumitem}

\usepackage[multiple]{footmisc}
\usepackage{mathrsfs}
\usepackage{todonotes}
\usepackage{tikz}
\usepackage{hyperref}
\usepackage{cleveref}
\usepackage{algpseudocode,algorithm,algorithmicx}
\usepackage{xfrac}

\usepackage{graphicx} 
\usepackage{color}

\usepackage{array}
\usepackage{amssymb}
\usepackage{amsmath}
\usepackage{xspace}
\usepackage{fancyhdr}
\usepackage{comment}

\hypersetup{final}

\newcommand{\boldzero}{\ensuremath{\boldsymbol{0}}}

\newcommand{\bfc}{\ensuremath{\mathbf{c}}}

\newcommand{\bfg}{\ensuremath{\mathbf{g}}}

\newcommand{\calD}{\ensuremath{\mathcal{D}}}

\newcommand{\calM}{\ensuremath{\mathcal{M}}}

\newtheorem{lem}{Lemma}[section]

\newtheorem{thm}[lem]{Theorem}

\newtheorem{defn}[lem]{Definition}

\newtheorem{obs}[lem]{Observation}

\makeatletter
\newcommand{\vast}{\bBigg@{4}}
\newcommand{\Vast}{\bBigg@{5}}
\makeatother

\newcommand{\ex}[2]{{\ifx&#1& \mathbb{E} \else
\underset{#1}{\mathbb{E}} \fi \left[#2\right]}}
\newcommand{\pr}[2]{{\ifx&#1& \mathbb{P} \else
\underset{#1}{\mathbb{P}} \fi \left[#2\right]}}

\newcommand{\ltwo}[1]{\left\|#1\right\|_2}

\usepackage{ulem}

\DeclarePairedDelimiterX{\infdivx}[2]{(}{)}{%
  #1\;\delimsize\|\;#2%
}

\renewcommand{\epsilon}{\varepsilon}

\newcommand{\clip}[2]{{\sf clip}\left(#1,#2\right)}

\setlist{nolistsep}
\setlist[itemize]{noitemsep, topsep=0pt}

\setlist{nolistsep}
\setlist[itemize]{noitemsep, topsep=0pt}

\usepackage{fullpage}
\usepackage[numbers, sort, comma, square]{natbib}

\begin{document}

\title{Tight Group-Level DP Guarantees for DP-SGD with Sampling via Mixture of Gaussians Mechanisms}
\author{Arun Ganesh\thanks{Google Research, \texttt{arunganesh@google.com}}}
\maketitle

\begin{abstract}
We give a procedure for computing group-level $(\epsilon, \delta)$-DP guarantees for DP-SGD, when using Poisson sampling or fixed batch size sampling. Up to discretization errors in the implementation, the DP guarantees computed by this procedure are tight (assuming we release every intermediate iterate).
\end{abstract}

\section{Introduction}

We consider DP-SGD with sampling: We have a dataset $D$ and a loss function $\ell$, and start with an initial model $\theta_0$. In round $i$, we sample a subset of $D$, $D_i$, according to some distribution (in this text, we will consider Poisson sampling and fixed batch size sampling). Let $\clip{\bfg}{L} := \frac{\bfg}{\max\{1, \ltwo{\bfg}/L\}}$. For a given noise multiplier $\sigma$, clip norm $L$, and (wlog, constant) step size $\eta$, DP-SGD is given by $T$ iterations of the update

\begin{equation}\label{eq:dp-sgd}
\theta_i = \theta_{i-1} - \eta \left(\sum_{d \in D_i} \clip{\nabla \ell(\theta_{i-1}, d)}{L} + \xi_i\right), \xi_i \stackrel{i.i.d.}{\sim} N(\boldzero, \sigma^2 L^2 \cdot \mathbb{I}).
\end{equation}

We are interested in proving group-level $(\epsilon, \delta)$-DP guarantees for DP-SGD with sampling, i.e., DP guarantees when we define two datasets $D$ and $D'$ to be adjacent if $D$ is equal to $D'$ plus up to $k$ added examples (or vice-versa). To the best 
of our knowledge, the best way to do this prior to this note was to prove an example-level DP guarantee and then use the following lemma of \cite{vadhan2017complexity}:

\begin{lem}[Lemma 2.2 in \cite{vadhan2017complexity}]\label{lem:vadhan}
If a mechanism satisfies $(\epsilon, \delta)$-DP with respect to examples, it satisfies $(k \epsilon, k e^{k \epsilon} \delta)$-DP with respect to groups of size $k$.
\end{lem}

While we can compute tight example-level DP guarantees using privacy loss distribution accounting \cite{koskela21tight, gopi21numerical, doroshenko2022connect}, using \cref{lem:vadhan} to turn these into group-level DP guarantees has two weaknesses. First, while \cref{lem:vadhan} is powerful in that it applies to any $(\epsilon, \delta)$-DP mechanism in a black-box manner, since it does not use the full $(\epsilon, \delta)$ tradeoff curve (i.e., it is agnostic to details of the mechanism) it is unlikely to be tight for any specific mechanism such as DP-SGD. Second, in order to prove a mechanism satisfies $(\epsilon, \delta)$-DP with respect to groups of size $k$, we need to show it satisfies $(\epsilon', \delta')$-DP with respect to one example, where $\delta'$ can potentially be much smaller than $\delta$ for large $k$ and $\epsilon'$. This is problematic in practice, as for typical choices of $\delta$, $\delta'$ might be small enough to cause numerical stability issues. Indeed, at the time of writing TensorFlow Privacy's \texttt{compute\_dp\_sgd\_privacy\_lib}\footnote{\url{https://github.com/tensorflow/privacy/blob/master/tensorflow_privacy/privacy/analysis/compute_dp_sgd_privacy_lib.py}} will sometimes return $\epsilon = \infty$ when computing group-level privacy guarantees because of this numerical stability issue. For example, for 2000 rounds of DP-SGD, with Poisson sampling with probability 1/100, $\sigma = 1$, $\delta = 10^{-6}$, \texttt{compute\_dp\_sgd\_privacy\_lib} reports $\epsilon = \infty$ with respect to groups of $k \geq 9$ (see \cref{sec:empirical}).

In this note, we give an alternate analysis using a tool called Mixture of Gaussians (MoG) Mechanisms developed in \cite{choquettechoo2023privacy} that addresses both these issues: The group-level DP guarantees it computes are tight, and it avoids the need to compute $\epsilon$ for small values of $\delta$ and the associated numerical stability errors.

This note is structured as follows: In \cref{sec:pld} we give a quick summary of privacy loss distribution accounting, including discussing accounting for mixture of Gaussians mechanisms. In \cref{sec:poisson} we state our approach for computing tight group-level DP guarantees in the cases of Poisson and fixed batch size sampling, including a brief argument that both approaches are tight in general. In \cref{sec:empirical}, we give an empirical comparison between our approach and the approach using \cref{lem:vadhan}. 
\section{Privacy Loss Distributions}\label{sec:pld}

Here we give a brief overview of key definitions for privacy loss distribution (PLD) accounting, and then at the end describe the overall high-level strategy of PLD accounting that we will employ.

\subsection{Definitions and Lemmas}

We formalize the notion of adjacent datasets as follows:

\begin{defn}
Given a data domain $\calD$, an \textbf{adjacency} $A$ is a set of ordered pairs of databases satisfying some relation, i.e. is a subset of $\calD^* \times \calD^*$.
\end{defn}

Approximate DP is closely tied to the $\alpha$-hockey-stick divergence:
\begin{defn}
For $\alpha \geq 0$, the \textbf{$\alpha$-hockey-stick divergence} between $P$ and $Q$ is \[H_\alpha(P, Q) := \int \max\left\{P(x) - \alpha \cdot Q(x), 0\right\} dx = \mathbb{E}_{x \sim P}\left[ \max\left\{1 - \alpha \cdot \frac{Q(x)}{P(x)}, 0\right\}\right].\]
\end{defn}

An $(\epsilon, \delta)$-DP guarantee under adjacency $A$ is equivalent to the statement:

\[\forall (D, D') \in A: H_{e^\epsilon}(\calM(D), \calM(D')) \leq \delta.\]

\begin{defn}
The \textbf{privacy loss distribution (PLD)} of $\calM(D)$ and $\calM(D')$ is the distribution of $\ln \left(\frac{\calM(D)(x)}{\calM(D')(x)}\right)$ where $x$ has distribution $\calM(D)$.
\end{defn}

From the definition of $\alpha$-hockey-stick divergence, we see that knowing the PLD of $\calM(D)$ and $\calM(D')$ suffices to compute the $\alpha$-hockey-stick divergence between $\calM(D)$ and $\calM(D')$ for all $\alpha \geq 0$. In particular, if $L$ is the privacy loss random variable of $\calM(D)$ and $\calM(D')$, then:

\begin{equation}\label{eq:pld-to-dp}
    H_\alpha(\calM(D), \calM(D')) = \mathbb{E}_{L}\left[ \max\left\{1 - \alpha \cdot e^{-L}, 0\right\}\right].
\end{equation}

Dominating pairs can be used to characterize (an upper bound on) the worst-case DP guarantees of a mechanism over all adjacent pairs of databases:

\begin{defn}[Definition 7 of \cite{zhucharacteristic2022}]
We say $P, Q$ \textbf{dominates} $P', Q'$ if for all $\alpha \geq 0$:

\[H_\alpha(P', Q') \leq H_\alpha(P, Q).\]

$P, Q$ are a \textbf{dominating pair for $\calM$ under adjacency $A$} if for all $(D, D') \in A$, $P, Q$ dominates $\calM(D), \calM(D')$.
\end{defn}

In turn, to prove DP guarantees for a mechanism it suffices to compute the PLD for a dominating pair for the mechanism rather than compute the PLD of $\calM(D), \calM(D')$ for all pairs of adjacent databases. We will restrict our attention to the add-up-to-$k$ and remove-up-to-$k$ adjacency, which correspond to adding or removing (at most) $k$ examples from $D$ to get $D'$, i.e. capture the notion of adjacency for group-level privacy of groups of size at most $k$:

\begin{defn}
Given a data domain $\calD$, the \textbf{add-up-to-$k$ adjacency} is $A_{add, k} := \{(D, D') \in \calD^* \times \calD^* : D \subseteq D', |D| + k \geq |D'|\}$. The \textbf{remove-up-to-$k$ adjacency} is $A_{rem, k} := \{(D, D') \in \calD^* \times \calD^*: D \supseteq D', |D| \leq |D'| + k\}$. The \textbf{add-or-remove-up-to-$k$ adjacency} is $A_{k} := A_{add, k} \cup A_{rem, k}$.
\end{defn}

Note that $k = 1$ retrieves a standard notion of example-level DP.
The last statement we need is the fact that dominating pairs can be adaptively composed:

\begin{lem}[Theorem 10 in \cite{zhucharacteristic2022}]\label{lem:adaptive-composition}
Let $\calM_1, \calM_2$ be two mechanisms. Under any adjacency, if $P_1, Q_1$ is a dominating pair for $\calM_1$ and $P_2, Q_2$ is a dominating pair for $\calM_2$, then $P_1 \times P_2, Q_1 \times Q_2$ is a dominating pair for $\calM_1 \times \calM_2$. Furthermore, this holds if $P_2, Q_2$ are fixed whereas $\calM_2$ is chosen adaptively based on the output of $\calM_1$, as long as $P_2, Q_2$ is a dominating pair for all possible choices of $\calM_2$.
\end{lem}

While the above lemma is stated for pairs of mechanisms, it is straightforward to extend it to a sequence of dominating pairs $(P_1, Q_1), (P_2, Q_2) \ldots, (P_n, Q_n)$ for an adaptively chosen mechanisms $\calM_1, \calM_2, \ldots, \calM_n$ of arbitrary length. Here, the requirement is that $(P_i, Q_i)$ is a dominating pair for all possible choices of $\calM_i$. We can compute the PLD of a composition of mechanisms from the individual mechanisms' PLDs via the following observation:

\begin{obs}
Let $L_1$ be the privacy loss distribution of $P_1, Q_1$ and $L_2$ be the privacy loss distribution of $P_2, Q_2$. Then $L_1 \ast L_2$, where $\ast$ denotes convolution, gives the privacy loss distribution of $P_1 \times P_2, Q_1 \times Q_2$.
\end{obs}

\subsection{High-Level Strategy}

Combining the definitions and lemmas given so far, we can arrive at the following high-level strategy, PLD accounting, for computing valid (and often tight) $(\epsilon, \delta)$-DP guarantees for a sequence of adaptively chosen mechanisms $\calM_1, \calM_2, \ldots, \calM_n$ under the add-or-remove-up-to-$k$ adjacency. The validity of this procedure follows from these definitions and lemmas and is covered in detail in the PLD accounting literature (see e.g. \cite{koskela21tight}), so we do not give a formal proof here.

\begin{enumerate}
    \item Choose $(P_1, Q_1), (P_2, Q_2), \ldots, (P_n, Q_n)$ that are dominating pairs for (all possible adaptive choices of) $\calM_1, \calM_2, \ldots, \calM_n$ under $A_{add, k}$.
    \item Compute $L_1, L_2, \ldots, L_n$, the privacy loss distributions of $(P_1, Q_1), (P_2, Q_2), \ldots, (P_n, Q_n)$.
    \item Let $L_{add} = L_1 \ast L_2 \ast \ldots \ast L_n$.
    \item Repeat the previous steps under $A_{rem, k}$ to get $L_{rem}$
    \item If \[\mathbb{E}_{L_{add}}\left[ \max\left\{1 - e^{\epsilon-L_{add}}, 0\right\}\right] \leq \delta, \mathbb{E}_{L_{rem}}\left[ \max\left\{1 - e^{\epsilon-L_{rem}}, 0\right\}\right] \leq \delta,\] then we can report that $\calM$ satisfies $(\epsilon, \delta)$-DP under the add-or-remove-up-to-$k$ adjacency. In particular, for a target $\delta$ we can compute the minimum $\epsilon$ such that these inequalities hold (or vice-versa).
\end{enumerate}
\subsection{Mixture of Gaussians Mechanisms}\label{sec:mog}

Mixture of Gaussians mechanisms are an analytic tool introduced in \cite{choquettechoo2023privacy}, representing a generalization of subsampled Gaussian mechanisms where the sensitivity is a random variable.

\begin{defn}
A \textbf{mixture of Gaussians (MoG) mechanism}, $\calM_{MoG}(\{\bfc_i\}, \{p_i\})$, is defined by a sensitivity vector random variable, i.e. a list of sensitivities $\{\bfc_1, \bfc_2, \ldots, \bfc_n\}$ and a list of probabilities of the same length $\{p_1, p_2, \ldots, p_n\}$ whose sum is 1, and standard deviation $\sigma$. Under the add adjacency, given $D$ the mechanism outputs a sample from $N(\boldzero, \sigma^2 \cdot \mathbb{I})$ and given $D'$ it samples from the same normal centered at the sensitivity random variable, i.e. it samples $i$ according to the probability distribution given by $\{p_i\}$ and outputs a sample from $N(\bfc_i, \sigma^2 \cdot \mathbb{I})$. Under the remove adjacency it is defined symmetrically.
\end{defn}

Lemma 4.6 in \cite{choquettechoo2023privacy} proves the following ''vector-to-scalar reduction'' statement about dominating pairs for MoG mechanisms:

\begin{lem}\label{lem:vector-to-scalar}
Let $\bfc, c$ be a vector and scalar random variable such that the random variable $\ltwo{\bfc}$ is stochastically dominated by $c$. Then under both the add and remove adjacencies, the MoG mechanism with random sensitivity $\bfc$ is dominated by the MoG mechanism with random sensitivity $c$.
\end{lem}

This lemma is useful because (1) we can easily compute the CDF of the privacy loss random variable for scalar MoG mechanisms with non-negative sensitivities (see \cite{choquettechoo2023privacy} for a detailed discussion), and thus to perform exact (up to numerical approximation) PLD accounting for these MoG mechanisms, and (2) as we will see in \cref{sec:poisson}, the application of this lemma to gradient queries is tight in general, i.e. for some instantiation of the gradient query the ``vector-to-scalar reduction'' preserves the privacy loss distribution (rather than worsening it). PLD accounting for (scalar) MoG mechanisms is currently supported by the open-source \texttt{dp\_accounting} \cite{pldlib} Python library\footnote{In particular, this accounting can be done using the privacy accountant \texttt{PLDAccountant} and the ``DP event'' \texttt{MixtureOfGaussiansDpEvent}. See \url{https://github.com/google/differential-privacy/blob/main/python/dp_accounting/dp_accounting/pld/pld_privacy_accountant.py} and \url{https://github.com/google/differential-privacy/blob/main/python/dp_accounting/dp_accounting/dp_event.py} respectively.}.
\section{Reducing DP-SGD with Sampling to MoG Mechanisms}\label{sec:poisson}

\subsection{Poisson Sampling}

Recall that in Poisson sampling $D_i$ in \eqref{eq:dp-sgd} is chosen by including each example in $D$ independently with some probability $q$. In this section we prove the following theorem on group-level DP guarantees for DP-SGD:

\begin{thm}\label{thm:poisson}
The output distributions of the composition of $T$ scalar MoG mechanisms under the add (resp. remove) adjacency with random sensitivity $Binom(k, q)$ and standard deviation $\sigma$ are a dominating pair for $T$-round DP-SGD with noise multiplier $\sigma$ and Poisson sampling with sampling probability $q$ under $A_{add, k}$ (resp. $A_{rem, k}$).
\end{thm}

Via the tools described in \cref{sec:mog}, this implies a procedure for computing tight group-level DP guarantees for DP-SGD with Poisson sampling: compute the PLD for a single MoG mechanism with the given sensitivities and probabilities, convolve it with itself $T$ times, and convert the resulting PLD to a $(\epsilon, \delta)$-DP guarantee.

\begin{proof}
For simplicity of presentation, we will assume all gradients have norm at most $L$, i.e. the clipping operator is the identity and can be dropped from \eqref{eq:dp-sgd}. We focus only on the add-up-to-$k$ adjacency $A_{add, k}$, a proof for $A_{rem, k}$ follows symmetrically.

We will show that in round $i$, for any fixed $\theta_{i-1}$, the pair of output distributions of $\theta_i$ given $D$ and $D'$ is dominated by the output distributions of a scalar MoG mechanism. Since this domination statement holds for any (adaptively chosen) $\theta_{i-1}$, the theorem holds by \cref{lem:adaptive-composition}.

Let $G$ be the group of at most $k$ sensitive examples, i.e. $D' = D \cup G$ and $D \cap G = \emptyset$. The distribution of $D \cap D_i$, i.e. the subset of $D$ included in $D_i$, is the same whether we sample from $D$ or $D'$. We will use the following ``quasi-convexity'' property of $(\epsilon, \delta)$-DP to treat $D \cap D_i$ as a constant in the analysis.

\begin{lem}\label{lem:quasi-convexity}
Let $w_1, w_2, \ldots, w_n \geq 0$ be probabilities summing to 1. Given distributions $\{P_i\}, \{Q_i\}$, let $P = \sum_i w_i P_i$ and $Q = \sum_i w_i Q_i$. Then for any $\alpha \geq 0$:

\[H_\alpha(P, Q) \leq \max_i H_\alpha(P_i, Q_i).\]
\end{lem}
\begin{proof}
We have:

\begin{align*}
H_\alpha(P, Q) &= \int \max\{P(x) - \alpha Q(x), 0\} dx = \int \max\{\sum_i w_i (P_i(x) - \alpha Q_i(x)), 0\} dx\\
&\stackrel{(\ast_1)}{\leq} \int \sum_i w_i \max\{P_i(x) - \alpha Q_i(x), 0\} dx = \sum_i w_i \int \max\{P_i(x) - \alpha Q_i(x), 0\} dx\\
&= \sum_i w_i H_\alpha(P_i, Q_i) \stackrel{(\ast_2)}{\leq} \max_i H_\alpha(P_i, Q_i).\\
\end{align*}

$(\ast_1)$ is the observation that $\max\{a+c, b+d\} \leq \max\{a, b\} + \max\{c, d\}$ i.e. ``the max of sums is less than the sum of maxes'' and $(\ast_2)$ holds because the $w_i$ are non-negative and sum to 1.
\end{proof}

Let $\Theta_i, \Theta_i'$ denote the distribution of $\theta_i$ given $D$ and $D'$ respectively and let $\Theta_i | S, \Theta_i' | S$ denote these distributions conditioned on the event $D \cap D_i = S$. Then we can write $\Theta_i = \sum_S q^{|S|}(1-q)^{|D|-|S|} \cdot (\Theta_i | S)$ and $\Theta_i' = \sum_S q^{|S|}(1-q)^{|D|-|S|} \cdot (\Theta_i' | S)$. By \cref{lem:quasi-convexity}, we then have for all $\alpha$, $H_\alpha(\Theta_i, \Theta_i') \leq \max_S H_\alpha(\Theta_i | S, \Theta_i' | S)$. So we just need to show that for any $S$, the pair $\Theta_i | S, \Theta_i' | S$ is dominated by a scalar MoG mechanism. We can write the exact form of each distribution for a fixed choice of $S$:

\[\Theta_i | S = N\left(\theta_{i-1} - \eta \sum_{d \in S} \nabla \ell(\theta_{i-1}, d), \eta^2 \sigma^2 L^2 \cdot \mathbb{I}\right), \]
\[\Theta_i' | S = \sum_{H \subseteq G} q^{|H|} (1-q)^{|G| - |H|} N\left(\theta_{i-1} - \eta \sum_{d \in S} \nabla \ell(\theta_{i-1}, d) - \eta \sum_{d \in H} \nabla \ell(\theta_{i-1}, d), \eta^2 \sigma^2 L^2 \cdot \mathbb{I}\right). \]

Distinguishing these distributions is equivalent to distinguishing them after the linear transformation $f(x) = - \frac{x - \theta_{i-1} - \eta \sum_{d \in S} \nabla \ell(\theta_{i-1}, d)}{\eta L}$. The linear transformations have distributions:

\[f(\Theta_i | S) = N\left(\boldzero, \sigma^2 \cdot \mathbb{I}\right), \]
\[f(\Theta_i' | S) = \sum_{H \subseteq G} q^{|H|} (1-q)^{|G| - |H|} N\left(\sum_{d \in H} \nabla \ell(\theta_{i-1}, d) / L, \sigma^2 \cdot \mathbb{I}\right). \]

These are the output distributions of a vector MoG mechanism with sensitivities $\{\sum_{d \in H} \nabla \ell(\theta_{i-1}, d) / L\}_{H \subseteq G}$ and associated probabilities $\{q^{|H|} (1-q)^{|G| - |H|}\}_{H \subseteq G}$. By triangle inequality, we have $\ltwo{\sum_{d \in H} \nabla \ell(\theta_{i-1}, d) / L} \leq \sum_{d \in H} \ltwo{\nabla \ell(\theta_{i-1}, d) / L} \leq |H|$. $|H|$ is the random variable $Binom(|G|, q)$, which is stochastically dominated by $Binom(k, q)$. So, by \cref{lem:vector-to-scalar}, for all $S$ the pair $\theta_i | S, \theta_i' | S$ (and thus the pair $\theta_i, \theta_i'$) is dominated by the scalar MoG mechanism with sensitivities and probabilities defined by $Binom(k, q)$, completing the proof.
\end{proof}

\subsection{Fixed Batch Size Sampling}

While Poisson sampling is easier to analyze, it is more common to use a fixed batch size. Here we provide an analysis for fixed batch size sampling, i.e. for some fixed batch size $B$, $D_i$ is a uniformly random subset of $D$ (or $D'$) of size $B$.

Let $Hypergeom(B, n+k, k)$ denote the hypergeometric random variable equal to the distribution of the number of black balls drawn from a bag with $k$ black balls and $n$ white balls if we draw $B$ balls without replacement. We show the following:

\begin{thm}\label{thm:minibatch}
The output distributions of the composition of $T$ scalar MoG mechanisms under the add (resp. remove) adjacency with sensitivities and probabilities defined by 2 times $Hypergeom(B, n+k, k)$ and standard deviation $\sigma$ are a dominating pair for $T$-round DP-SGD with noise multiplier $\sigma$ and fixed batch size sampling with batch size $B$ on datasets of size at least $n$ under $A_{add, k}$ (resp. under $A_{rem, k}$).
\end{thm}

\begin{proof}
Similarly to \cref{thm:poisson} it suffices to show for any fixed $\theta_{i-1}$, the output distributions of $\theta_i$ given $D$ and $D'$ are dominated by scalar MoG mechanism with sensitivities and probabilities defined by 2 times $Hypergeom(B, n+k, k)$. Again for simplicity we focus on $A_{add,k}$ as $A_{rem,k}$ follows by a symmetric argument.

Again let $G$ be the group of at most $k$ sensitive examples, i.e. $D' = D \cup G$ and $D \cap G = \emptyset$. 
Also again let $\Theta_i, \Theta_i'$ denote the distribution of $\theta_i$ given $D$ and $D'$ respectively.
We will again evoke \cref{lem:quasi-convexity} by writing $\Theta_i$ and $\Theta_i'$ as mixtures. However, unlike for Poisson sampling the distribution of $D \cap D_i$ is not the same when sampling from $D$ as when sampling from $D'$, so $D \cap D_i$ is not the right object to condition on in defining the mixture. Instead, consider the following process that samples a uniformly random subset of $B$ examples from $D'$:

\begin{enumerate}
    \item Sample $S = {s_1, s_2, \ldots, s_B}$, a uniformly random ordered subset of $B$ examples from $D$.
    \item Sample $m \sim Hypergeom(B, |D|+|G|, |G|)$.
    \item Sample $H = {t_1, t_2, \ldots, t_m}$, a uniformly random ordered subset of $m$  examples from $G$, and replace the first $m$ examples of $S$ with $H$.
\end{enumerate}

The ordering of $S$ and $H$ is not necessary, but simplifies the presentation later. We use the above process to define the mixtures: let $\Theta_i | S$ denote $\Theta_i$ conditioned on $D_i = S$ (where $S$ is again ordered) and let $\Theta_i' | S$ denote $\Theta_i'$ conditioned on the choice of $S$ in step 1 in the above process. Then we can write $\Theta_i = \sum_S \frac{1}{\binom{n}{B} \cdot B!} (\Theta_i | S)$ and 
$\Theta_i' = \sum_S \frac{1}{\binom{n}{B} \cdot B!} (\Theta_i' | S)$, and thus apply \cref{lem:quasi-convexity} to conclude $H_\alpha(\Theta_i, \Theta_i') \leq \max_S H_\alpha(\Theta_i | S, \Theta_i' | S)$. It now suffices to show $\Theta_i | S, \Theta_i' | S$ are dominated by the MoG mechanism with sensitivities and probabilities defined by 2 times $Hypergeom(B, n+k, k)$.

For any $S$, using the same linear transformation $f$ as in \cref{thm:poisson}, we have:

\[f(\Theta_i | S) = N\left(\boldzero, \sigma^2 \cdot \mathbb{I}\right), \]
\[f(\Theta_i' | S) = \sum_{H \subseteq G} Pr[Hypergeom(B, |D|+|G|, |G|) = |H|] \cdot \frac{(|G|-|H|)!}{|G|!}\cdot\]
\[N\left(\sum_{j \in [|H|]} \left(\nabla \ell(\theta_{i-1}, t_j) - \nabla \ell(\theta_{i-1}, s_j) \right)/ L, \sigma^2 \cdot \mathbb{I}\right). \]

By triangle inequality we have $\sum_{j \in [|H|]} \left(\nabla \ell(\theta_{i-1}, t_j) - \nabla \ell(\theta_{i-1}, s_j) \right)/ L \leq 2 |H|$. $|H|$ is drawn from $Hypergeom(B, |D|+|G|, |G|)$. Since $|D| \geq n$ and $|G| \leq k$, $|H|$ is stochastically dominated by $Hypergeom(B, n+k, k)$, and so by \cref{lem:vector-to-scalar} the pair $\Theta_i | S, \Theta_i' | S$ is dominated by the MoG mechanism corresponding to with sensitivities and probabilities defined by 2 times $Hypergeom(B, n+k, k)$ as desired.
\end{proof}

\subsection{Tightness}

Both \cref{thm:poisson} and \cref{thm:minibatch} are easily seen to be tight, i.e. the $(\epsilon, \delta$)-tradeoff curve given by PLD analysis of the composition of MoG mechanisms is the exact $(\epsilon, \delta)$-DP characterization of an instance of DP-SGD on a certain loss function and pair of databases, assuming an adversary gets to view every iterate, i.e. every $\theta_i$. 

Let $L = 1$ for simplicity. In the case of Poisson sampling and \cref{thm:poisson}, consider the one-dimensional loss function which is $0$ for all examples except the $k$ sensitive examples, where the loss function is $\ell(\theta; d) = -\theta$, i.e. has gradient -1 everywhere. Then in each round the distributions of $\theta_i - \theta_{i-1}$ for $D$ and $D'$ are exactly the distributions of a MoG mechanism with random sensitivity $Binom(k, q)$ and the same noise standard deviation, i.e. the PLD of the composition of $T$ MoG mechanisms is an exact characterization of the privacy loss of DP-SGD for these database and loss functions.

For \cref{thm:minibatch}, a similar argument holds: Consider the one-dimensional loss function that is $\theta$ for all examples except for the $k$ sensitive examples, for which it is $-\theta$. Then the distributions of $\theta_i - \theta_{i-1} + B$ for $D$ and $D'$ are exactly the distributions of a MoG mechanism with random sensitivity $Hypergeom(B, n+k, k)$ under the add-up-to-$k$ adjacency, i.e. again the PLD of the composition of $T$ MoG mechanisms is an exact characterization of the privacy loss random variable.

The only possible improvement in the analysis (without further assumptions on the loss function / data distribution, e.g. convexity or using projection into a constraint set) is thus to use a ``last-iterate'' analysis, i.e. analyze the privacy guarantees of outputting only $\theta_T$ instead of all $\theta_i$. We conjecture for general non-convex losses that the improvement from ``every-iterate'' to ``last-iterate'' is marginal, and may actually be zero. To the best of our knowledge, all privacy analyses of DP-SGD that improve by using a last-iterate analysis (e.g. \cite{feldman2018privacy, ChourasiaYS21, altschuler2022privacy}) require further assumptions such as convexity, so disproving this conjecture would likely require substantially novel techniques and insights.
\section{Empirical Results}\label{sec:empirical}

In this section we plot the $\epsilon$ values computed empirically using \cref{thm:poisson} and \cref{thm:minibatch} in a common benchmark setting, and give a comparison to those computed using the conversion of \cref{lem:vadhan} and a naive lower bound.

We consider the setting of training on CIFAR-10 for 20 epochs, each epoch taking 100 iterations, as considered in e.g. \cite{choquette-choo2023amplified}. This corresponds to $T = 2000$ rounds of DP-SGD with sampling probability $q = .01$. With fixed batch size sampling specifically, this corresponds to $B = 500$ since there are $n = 50000$ training examples. We set our target $\delta$ to be $10^{-6}$, vary the noise multiplier $\sigma$ and group size $k$, and plot (1) the $\epsilon$ computed by PLD accounting for MoG mechanisms (``Our analysis'') and by (2) converting  an exact example-level DP guarantee to a group-level DP guarantee via \cref{lem:vadhan} (``Vadhan analysis''). We use the \texttt{PLDAccountant} class in \texttt{dp\_accounting} for (1) and the \texttt{tensorflow\_privacy} implementation of (2). We can also consider (3) a linear lower bound of $k \epsilon_1$, where $\epsilon_1$ is the privacy parameter for a single example at the same $\delta$ (``Lower bound''), to get a sense for how quickly $\epsilon$ computed by each approach grows relative to linear growth. In \cref{fig:poisson}, for each $\sigma$ we plot $\epsilon$ as a function of $k$ for each approach when using Poisson sampling. In \cref{fig:minibatch} for completeness we do the same but for fixed batch size sampling and a doubled noise multiplier, although the two settings have nearly the same $\epsilon$ as a function of $k$ (if we double the noise multiplier for fixed batch size sampling) since we have $n \gg B \gg k$, i.e. $Binom(n, k) \approx Hypergeom(B, n+k, k)$. 

\begin{figure}
    \centering
    \includegraphics[width=.3\linewidth]{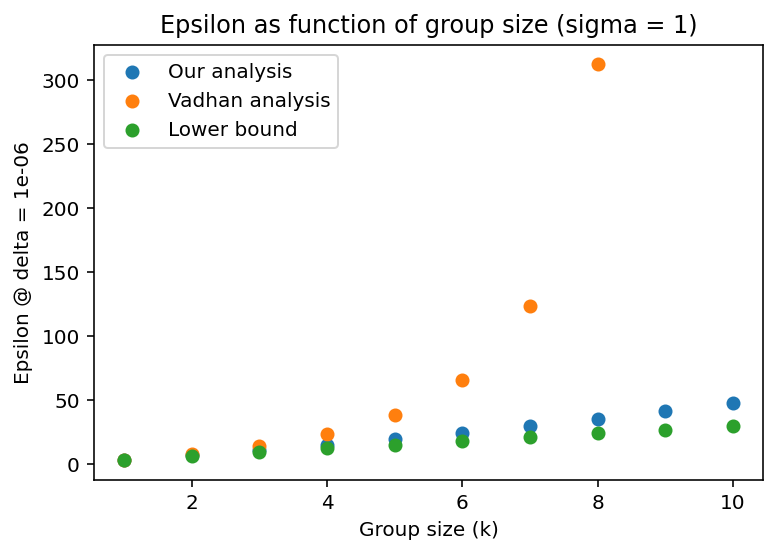}
    \includegraphics[width=.3\linewidth]{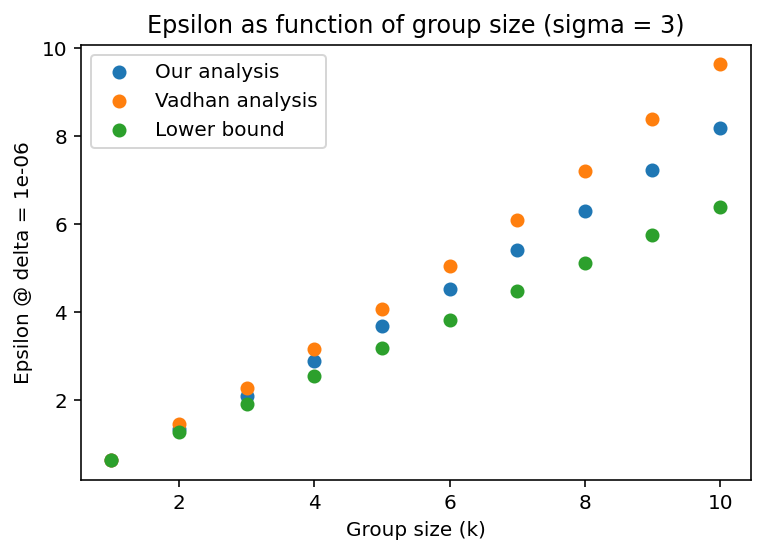}
    \includegraphics[width=.3\linewidth]{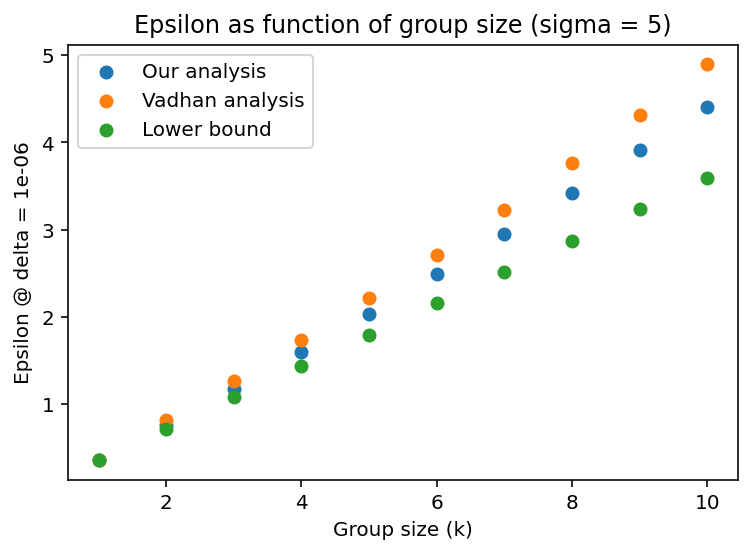}
    \caption{$\epsilon$ as a function of $k$ for Poisson sampling. We use $T = 2000, q = 1/100$. ``Our analysis'' is $\epsilon$ computed using the PLD of the dominating pair in \cref{thm:poisson}. ``Vadhan analysis'' computes an example-level DP guarantee and uses \cref{lem:vadhan}. ``Lower bound'' takes the example-level  $(\epsilon_1, \delta)$-DP guarantee and multiplies $\epsilon_1$ by $k$ to get a lower bound on the true $\epsilon$.}
    \label{fig:poisson}
\end{figure}

\begin{figure}
    \centering
    \includegraphics[width=.3\linewidth]{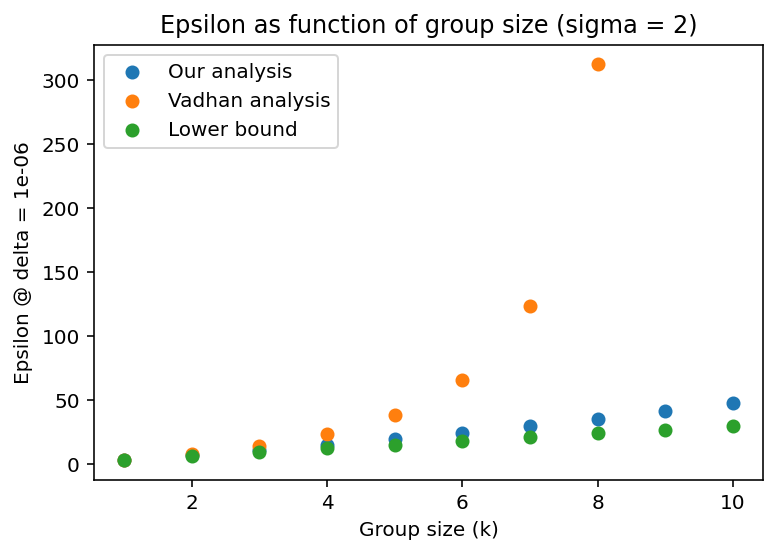}
    \includegraphics[width=.3\linewidth]{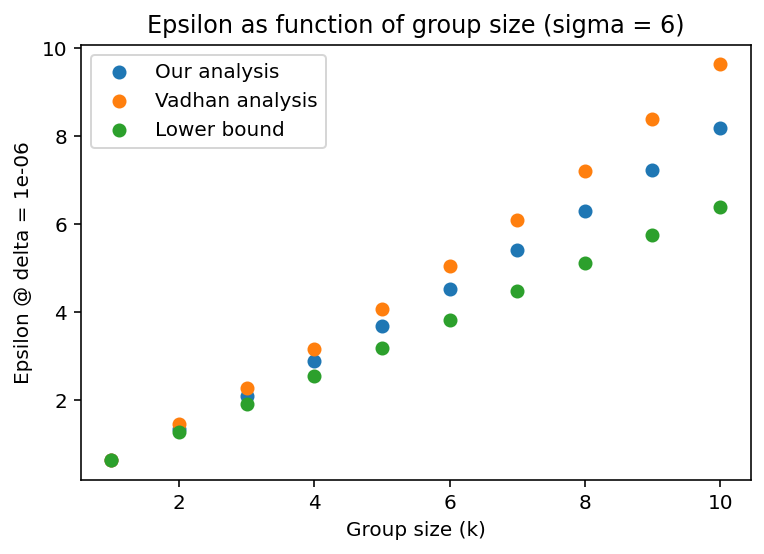}
    \includegraphics[width=.3\linewidth]{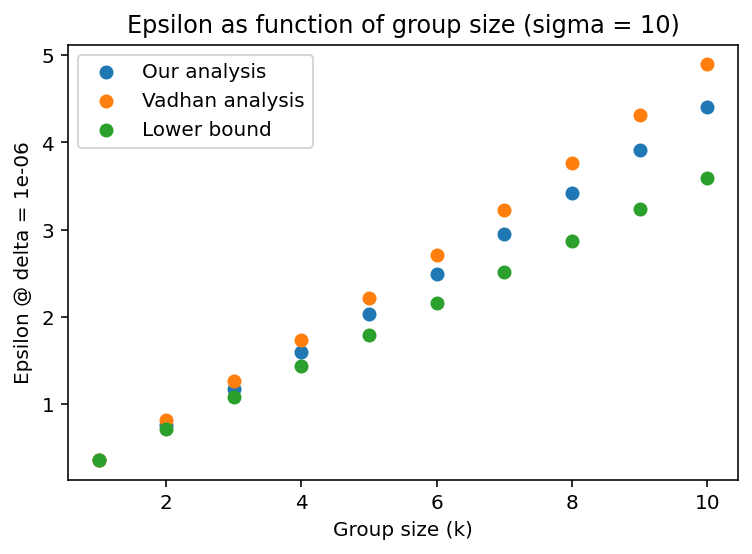}
    \caption{$\epsilon$ as a function of $k$ for fixed batch size sampling. We use $T = 2000, B = 500, n = 50000$. The labels are defined analogously to  \cref{fig:poisson}.}
    \label{fig:minibatch}
\end{figure}

From \cref{fig:poisson}, we see that as expected because our approach is tight, it consistently improves over using \cref{lem:vadhan}. From comparing to the lower bound, it is also clear that (perhaps surprisingly), the growth of the $\epsilon$ parameter of DP-SGD as a function of $k$ remains close to linear even in a regime where $\epsilon \approx 50$, whereas using \cref{lem:vadhan} results in $\epsilon$ that grows exponentially in the same regime. The relative improvement over using \cref{lem:vadhan} becomes larger as $k$ increases and as $\sigma$ decreases (i.e., as $\epsilon$ increases). Furthermore, for $\sigma = 1$ we see that the numerical instability due to using \cref{lem:vadhan} appears even for a moderate group size of $k = 9$ and causes the $\epsilon$ computed by the methods in \texttt{tensorflow\_privacy} to be infinite, whereas our approach is able to compute an $\epsilon$ close to the linear lower bound for this group size.

\section*{Acknowledgements}

The author is thankful to Christopher Choquette-Choo for co-authoring of the initial codebase of \cite{choquettechoo2023privacy} that was built upon in the process of writing this note, to Pritish Kamath for help with integration of MoG mechanisms into the \texttt{dp\_accounting} library, and to Galen Andrew who made the author aware of the question of getting better user-level DP guarantees for DP-SGD. 

\bibliographystyle{alpha}
\bibliography{ref}
\end{document}